\documentclass[a4paper,UKenglish]{lipics-v2016}
\usepackage{amsmath,amsthm}

\newtheorem{claim}{Claim}
\newtheorem{proposition}{Proposition}
\usepackage{amssymb,calc}
\usepackage{mdframed}
\usepackage{todo}
\usepackage{microtype}%if unwanted, comment out or use option "draft"
\title{Faster Exact and Parameterized Algorithm for Feedback Vertex Set in Bipartite Tournaments}
\titlerunning{Faster Exact and Parameterized Algorithm for Feedback Vertex Set in Tournaments} %optional, in case that the title is too long; the running title should fit into the top page column

\author[1]{Mithilesh Kumar}
\author[2]{Daniel Lokshtanov}
\affil[1]{Department of Informatics, University of Bergen\\
  Norway\\
  \texttt{Mithilesh.Kumar@ii.uib.no}}
\affil[2]{Department of Informatics, University of Bergen\\
  Norway\\
  \texttt{daniello@ii.uib.no}}
\authorrunning{M. Kumar and D. Lokshtanov} %mandatory. First: Use abbreviated first/middle names. Second (only in severe cases): Use first author plus 'et. al.'

\Copyright{Mithilesh Kumar and Daniel Lokshtanov}%mandatory, please use full first names. LIPIcs license is "CC-BY";  http://creativecommons.org/licenses/by/3.0/

% mandatory: Please choose ACM 1998 classifications from http://www.acm.org/about/class/ccs98-html . E.g., cite as "F.1.1 Models of Computation".
% mandatory: Please provide 1-5 keywords
\EventEditors{Akash Lal, S. Akshay, Saket Saurabh, and Sandeep Sen}
\EventNoEds{4}
\EventLongTitle{36th IARCS Annual Conference on Foundations of Software Technology and Theoretical Computer Science (FSTTCS 2016)}
\EventShortTitle{FSTTCS 2016}
\EventAcronym{FSTTCS}
\EventYear{2016}
\EventDate{December 13--15, 2016}
\EventLocation{Chennai, India}
\EventLogo{}
\SeriesVolume{65}
\ArticleNo{}
\begin{document}

\maketitle
%%%%%%%%%%%%%%%%%%%%%%%%%%%%%%%%%%%%%%%%%%%%%%%%%%%%%%%%
%%%%%%%%%%%%%%%%%%%%%%%%%%%%%%%%%%%%%%%%%%%%%%%%%%%%%%%%
%%%%%%%%%%%%%%%%%%%%%%%%%%%%%%%%%%%%%%%%%%%%%%%%%%%%%%%%
\begin{abstract}

A {\em bipartite tournament} is a directed graph $T:=(A \cup B, E)$ such that every pair of vertices $(a,b), a\in A,b\in B$ are connected by an arc, and no arc connects two vertices of $A$ or two vertices of $B$. A {\em feedback vertex set} is a set $S$ of vertices in $T$ such that $T - S$ is acyclic. In this article we consider the {\sc Feedback Vertex Set} problem in bipartite tournaments. Here the input is a bipartite tournament $T$ on $n$ vertices together with an integer $k$, and the task is to determine whether $T$ has a feedback vertex set of size at most $k$. We give a new algorithm for {\sc Feedback Vertex Set in Bipartite Tournaments}. The running time of our algorithm is upper-bounded by $O(1.6181^k + n^{O(1)})$, improving over the previously best known algorithm with running time $2^kk^{O(1)} + n^{O(1)}$ [Hsiao, ISAAC 2011]. As a by-product, we also obtain the fastest currently known exact exponential-time algorithm for the problem, with running time $O(1.3820^n)$. 
\end{abstract}
%------------------------------------------------------------
%%%%%%%%%%%%%%%%%%%%%%%%%%%%%%%%%%%%%%%%%%%%%%%%%%%%%%%%
%%%%%%%%%%%%%%%%%%%%%%%%%%%%%%%%%%%%%%%%%%%%%%%%%%%%%%%%
%%%%%%%%%%%%%%%%%%%%%%%%%%%%%%%%%%%%%%%%%%%%%%%%%%%%%%%%
\section{Introduction}

A {\em feedback vertex set} in a graph $G$ is a vertex set whose removal makes the graph acyclic. The {\sc Feedback Vertex Set} problem is a well-studied graph problem where input is a graph $G$ (directed or undirected) and the task is to find a smallest possible feedback vertex set. Finding such an optimal feedback vertex set turns out to be \textsc{NP}-complete~\cite{GJ79}, indeed the problem is one of the very first to be shown NP-complete in the influential paper of Karp~\cite{Karp72}.
%not only on general undirected and directed graphs 
Since, polynomial time algorithms are highly unlikely, {\sc Feedback Vertex Set} on general directed and undirected graphs has been extensively studied from the perspective of approximation algorithms~\cite{BafnaBF99,EvenNSS98}, parameterized algorithms~\cite{ChenLLOR08, CyganNPPRW11, KociumakaP14}, exact exponential-time algorithms~\cite{Razgon07,XiaoN15} as well as graph theory~\cite{erdHos1965independent,reed1996packing}. 

This paper belongs to a long line of work studying the complexity of {\sc Feedback Vertex Set} on restricted classes of graphs. On one hand {\sc Feedback Vertex Set} remains NP-complete on tournaments and bipartite tournaments \cite{CaiDZ02}, planar undirected graphs~\cite{GJ79}, planar directed graphs with in-degree and out-degree at most $3$~\cite{GJ79} as well as directed graphs with in-degree and out-degree at most $2$~\cite{GJ79}. On the other hand the problem is polynomial time solvable on undirected graphs of maximum degree $3$~\cite{UenoKG88}, chordal graphs~\cite{festa1999feedback} and weakly chordal graphs~\cite{FominV10}, indeed on any class of graphs with polynomially many potential maximal cliques~\cite{FominV10}. 
Being a problem of fundamental importance, {\sc Feedback Vertex Set} has been approached algorithmically even on the classes of graphs where it remains NP-complete. For example the problem admits (efficient) polynomial time approximation schemes~\cite{Cohen-AddadVKMM16,DemaineH05,FominLRS11}, sub-exponential time parameterized algorithms~\cite{DemaineFHT05} and linear kernels~\cite{FominLST10} on classes of graphs excluding a fixed graph $H$ as a minor. In this paper we study the problem on {\em bipartite tournaments}. 

A {\em tournament} is a subclass of directed graphs where every pair of vertices are connected by an arc. A {\em bipartite tournament} is a directed graph where the vertices are partitioned into two sets $A$ and $B$, there is an arc connecting every vertex in $A$ with every vertex in $B$, and there are no edges between vertices of $A$ and vertices of $B$. Tournaments arise naturally from round-robin competitions whereas bipartite tournaments model a two-team competition in which every player in one team plays against every player of the other team. Here arcs are drawn from the winning to the losing player, and often one seeks to rank the players from ``best'' to ``worst'' such that players that appear higher in the ranking beat all lower ranked players they played against. Such an absolute ranking possible only if there are no cycles in the tournament. The size of the smallest feedback vertex set then becomes a measure of how far the tournament is from admitting a consistent ranking. For this reason the structure of cycles and feedback vertex sets in (bipartite) tournaments has been studied both from the perspective of graph theory~\cite{BEINEKE1982140,CLAPHAM1985195,Harary} and algorithms.

For bipartite tournaments, finding a feedback vertex set reduces to hitting all cycles of length $4$. For this reason the {\sc Feedback Vertex Set} problem is more computationally tractable on bipartite tournaments than on general directed graphs. Specifically the best known approximation algorithm for {\sc Feedback Vertex Set} on directed graphs has an approximation factor of $O(\log n \cdot \log\log n)$~\cite{EvenNSS98}, and the problem does {\em not} admit a constant factor approximation assuming the Unique Games Conjecture~\cite{GuruswamiHMRC11}. On bipartite tournaments it is easy to obtain a $4$-approximation (see Lemma~\ref{approximation}). Further, an improved approximation algorithm with ratio $2$ was obtained by Zuylen.~\cite{Zuylen11}. 

Similarly, it was open for a long time whether {\sc Feedback Vertex Set} on general directed graphs admits an FPT algorithm, that is an algorithm that determines whether there exists a solution of size at most $k$ in time $f(k)n^{O(1)}$. In 2008, Chen et al. \cite{ChenLLOR08} gave an algorithm with running time $O(4^kk^{O(1)}k!nm)$, and it is an outstanding open problem whether there exists an algorithm with running time $2^{O(k)}n^{O(1)}$. For bipartite tournaments, the realization that it is necessary and sufficient to hit all cycles of length $4$ yields a simple $4^kn^{O(1)}$ time parameterized algorithm: recursively branch on vertices of a cycle of length $4$. Tru{\ss}~\cite{truss2005parameterized} gave an improved algorithm with running time $3.12^kn^{O(1)}$, Sasatte \cite{Sasatte08} further improved the running time to $3^kn^{O(1)}$, while Hsiao \cite{Hsiao11} gave an algorithm with running time $2^kn^{O(1)}$. Prior to this work, this was the fastest known parameterized algorithm for {\sc Feedback Vertex Set} on bipartite tournaments. Our main result is an algorithm with running time $O(1.6181^k + n^{O(1)})$. Using the recent black-box reduction from parameterized to exact exponential time algorithms of Fomin et al.~\cite{FominGLS16} we also obtain an exponential-time algorithm running in $O(1.3820^n)$ time.

\smallskip
\noindent
{\bf Methods.} 
Our algorithm is based on the recent parameterized algorithm with running time  $O(1.6181^k + n^{O(1)})$ by the authors~\cite{KumarL16} for {\sc Feedback Vertex Set} in tournaments.  The main idea of this algorithm is that tournaments are very {\em rigid}.  Given as input a tournament $T$, by obtaining a large set $M$ of vertices that is {\em disjoint} from the feedback vertex set $H$ sought for, we can get a rough sketch of the rigid structure of $T - H$. This structure is then very useful for recovering the solution $H$. Indeed, the only way that vertices that are ``far apart'' in the approximate sketch of the structure of $T - H$ can interact with each other is by being ``in conflict''. Out of two vertices that are in conflict, one of them has to be deleted. Thus, dealing with conflicts can be done in a similar fashion as with edges in the {\sc Vertex Cover} problem. For any vertex $v$ appearing in at least two conflicts, branch into two sub-problems. In the first sub-problem $v$ is deleted, in the second all vertices in conflict with $v$ are deleted. If there are {\em no} conflicts it is sufficient to solve the  {\sc Feedback Vertex Set}  problem ``locally''. If every vertex appears in at most one conflict a divide and conquer approach can be taken.

Because bipartite tournaments are also quite ``rigid'', we expected that the same approach would easily give an algorithm for {\sc Feedback Vertex Set} on bipartite tournaments with the same running time. Our expectations were both wrong and correct; indeed we {\em do} obtain an algorithm for  {\sc Feedback Vertex Set} on bipartite tournaments with the same template and the same running time as the algorithm for tournaments~\cite{KumarL16}, yet the adaptation turned out to be anything but easy. Specifically, in virtually every step of the algorithm, the lack of a unique topological sort of acyclic bipartite tournaments presented significant challenges.

The fact that these challenges still could be overcome by sub-exponential time cleaning procedures gives hope that the same template could be applicable in several situations where one seeks a ``small'' set of vertices or edges to delete in order to modify the input graph to a ``rigid'' structure; such as {\sc Cluster Vertex Deletion}, {\sc Cograph Vertex Deletion} and {\sc Feedback Vertex Set} in the more general setting when the input graph is a multi-partite tournament~\cite{GutinY08}.

\smallskip
\noindent
{\bf Organization of the paper.} In Section \ref{pre} we set up definitions and notation, and state a few useful preliminary results. The standard graph notation and parameterized complexity terminology is set up in the appendix. In Section~\ref{M} we define and prove some properties of $M$-sequence.  In Section~\ref{family} we define and give an algorithm for Constrained Feedback Vertex Set problem.  

%%%%%%%%%%%%%%%%%%%%%%%%%%%%%%%%%%%%%%%%%%%%%%%%%%%%%%%%
%%%%%%%%%%%%%%%%%%%%%%%%%%%%%%%%%%%%%%%%%%%%%%%%%%%%%%%%
%%%%%%%%%%%%%%%%%%%%%%%%%%%%%%%%%%%%%%%%%%%%%%%%%%%%%%%%
\section{Preliminaries}\label{pre}

\smallskip
\noindent
{\bf Preliminary Results.}
If a bipartite tournament is acyclic then it does not contain any squares. It is a well-known and basic fact that the converse is also true, see e.g.~\cite{Dom201076}.
\begin{lemma}\label{triangle}\cite{Dom201076}
A bipartite tournament is acyclic if and only if it contains no squares.\end{lemma}
Lemma~\ref{triangle} immediately gives rise to a folklore greedy $4$-approximation algorithm for \textsf{BTFVS}: as long as $T$ contains a square, delete all the vertices in this square.

\begin{lemma}[folklore]\label{approximation}
There is a polynomial time algorithm that given as input a bipartite tournament $T$ and integer $k$, either correctly concludes that $T$ has no feedback vertex set of size at most $k$ or outputs a feedback vertex set of size at most $4k$. 
\end{lemma}
In fact, \textsf{BTFVS} has a polynomial time factor $3.5$-approximation, due to Cai et al.~\cite{CaiDZ00}. However, the simpler algorithm from Lemma~\ref{approximation} is already suitable to our needs.
The preliminary phase of our algorithm for \textsf{BTFVS} is the kernel of Dom et al.~\cite{Dom201076}. We will need some additional properties of this kernel that we state here. Essentially, Lemma~\ref{kernel} allows us to focus on the case when the number of vertices in the input bipartite tournament is $O(k^3)$.

\begin{lemma}\label{kernel}\cite{Dom201076} There is a polynomial time algorithm that given as input a bipartite tournament $T$ and integer $k$, runs in polynomial time and outputs a bipartite tournament $T'$ and integer $k'$ such that $|V(T')| \leq |V(T)|$, $|V(T')| = O(k^3)$, $k' \leq k$, and $T'$ has a feedback vertex set of size at most $k'$ if and only if $T$ has a feedback vertex set of size at most $k$.
\end{lemma}
For any sequence $\sigma$, let $|\sigma|$ denote the length of $\sigma$. For each $i=1,2,\dots,|\sigma|$, let $V_i$ be the $i$-th element of $\sigma$. Let $T$ be an $n$-vertex acyclic bipartite tournament. The \emph{canonical sequence} for $T$ is the sequence $\sigma$ of vertex sets that can be obtained from $T$ in $O(n^2)$ time as follows: For each $i\geq 1$, let $V_i$ consist of the vertices without incoming edges in $T\setminus \bigcup_{j=1}^{i-1}V_j$.
\begin{lemma}\label{sequence}\cite{Hsiao11}
Let $T$ be an $n$-node acyclic bipartite tournament. Let $\sigma$ be the canonical sequence for $T$. The following statements hold. (i) $V_1,V_2,\dots, V_{|\sigma|}$ form a partition of $V(T)$. (ii) For each directed edge $(u,v)$ of $T$, the vertex set $V_i$ containing $u$ precedes the vertex set $V_j$ containing $v$ in the sequence (i.e. $i<j$). (iii) $A=\bigcup_{i\equiv 1 \mod 2}V_i$ and $B=\bigcup_{i\equiv 0 \mod 2}V_i$ are the partite sets of $T$.
\end{lemma}

\begin{definition}[$t$-wise independent]
	A family $H_{n,t,q}$ of functions from $[n]$ to $[q]$ is called a $t$-wise independent sample space if, for every $t$ positions $1<i_1<i_2<\cdots<i_t\leq n$, and every tuple $\alpha\in [q]^t$, we have
	$Pr(f(i_1),f(i_2),\dots,f(i_t))=\alpha=q^{-t}$
	where the function $f\in H_{n,t,q}$ is chosen uniformly at random.
\end{definition}
\begin{theorem}\cite{AlonBI86}\label{twise}
	There exists a $t$-wise independent sample space $H_{n,t,q}$ of size $O(n^t)$ and it can be constructed efficiently in time linear in the output size.
\end{theorem}

%%%%%%%%%%%%%%%%%%%%%%%%%%%%%%%%%%%%%%%%%%%%%%%%%%%%%%%%
%%%%%%%%%%%%%%%%%%%%%%%%%%%%%%%%%%%%%%%%%%%%%%%%%%%%%%%%
%%%%%%%%%%%%%%%%%%%%%%%%%%%%%%%%%%%%%%%%%%%%%%%%%%%%%%%%
\section{$M$-Sequence}\label{M}

First we extend the notion of the canonical sequence to general bipartite tournaments relative to a set $M$ of vertices.
\begin{definition}[$M$-equivalent]
	Given a directed graph $T$ and a subset $M\subseteq V(T)$, two vertices $u,v\in V(T)$ are said to be $M$-equivalent if $N^+(u)\cap M=N^+(v)\cap M$ and $N^-(u)\cap M=N^-(v)\cap M$.
\end{definition} 
\begin{definition}[$(M,X)$-equivalent]
	Let $T$ be a bipartite tournament and a subset $M\subseteq V(T)$ such that $T[M]$ is acyclic. Let $(X_1,X_2,\dots)$ be the canonical sequence of $T[M]$. For any set $X_i$ in the canonical sequence of $T[M]$ and any vertex $v\in V(T)$, $v$ is called $(M,X_i)$-equivalent if $v$ is $M$-equivalent to a vertex in $X_i$. 	
	\end{definition}
\begin{definition}[$(M,X)$-conflicting]
	Let $T$ be a bipartite tournament and a subset $M\subseteq V(T)$ such that $T[M]$ is acyclic. Let $(X_1,X_2,\dots)$ be the canonical sequence of $T[M]$. For any set $X_i$ in $(X_1,X_2,\dots)$ and for any vertex $v\in V(T)$, $v$ is called $(M,X_i)$-conflicting if
	\begin{itemize}
 		\item  $N^+(v)\cap X_i\neq \emptyset$ and $N^-(v)\cap X_i\neq \emptyset$,
  		\item for every $j<i$, $N^+(v)\cap X_j=\emptyset$ and for every $j>i$, $N^-(v)\cap X_j=\emptyset$. 
 	\end{itemize}
\end{definition}
 \begin{definition}[$M$-consistent]
 	Let $T$ be a directed graph and $M\subseteq V(T)$. $T$ is called $M$-consistent if for every vertex $v\in V(T)$ $T[M\cup \{v\}]$ is acyclic. 
 \end{definition}
 As a direct consequence of the above definitions, we have the following lemma.
 \begin{lemma}\label{conflicting}
 	Let $T$ be an $M$-consistent bipartite tournament for some subset $M\subseteq V(T)$. Let $(X_1,X_2,\dots,X_i,X_{i+1},\dots)$ be the canonical sequence of $T[M]$. Let $v\in V(T)$ be a $(M,X_i)$-conflicting vertex. Then, the canonical sequence of $T[M\cup \{v\}]$ is $(X_1,X_2,\dots,X_i',\{v\},X_i'',X_{i+1},\dots)$ where $X_i'\cup X_i''=X_i$ such that $X_i',X_i''\neq \emptyset$. 
 \end{lemma}
 \begin{definition}[$M$-universal]
 	Let $T$ be a bipartite tournament and a subset $M\subseteq V(T)$ such that $T[M]$ is acyclic. Let $(X_1,X_2,\dots)$ be the canonical sequence of $T[M]$. A vertex $v\in V(T)$ is called $M$-universal if  the following holds:
%%% SHORT
(i) $v$ is not $(M,X_i)$-equivalent for any $X_i$, (ii) $T[M\cup \{v\}]$ is acyclic. (iii) There exists a topological sort of $T[M\cup \{v\}]$ such that $v$ is either the first vertex (called $M^-$-universal) or is the last vertex (called $M^+$-universal) in the ordering. 	
 \end{definition}
\begin{lemma}\label{Mconsistent}
 	Let $T$ be an $M$-consistent bipartite tournament and let $(X_1,X_2,\dots)$ be the canonical sequence of $T[M]$. Then, for every vertex $v\in V(T)$, there exists a unique index $i$ such that $v$ satisfies exactly one of the following properties:
%% SHORT
(i) $v$ is $(M,X_i)$-equivalent, (ii) $v$ is $(M,X_i)$-conflicting, (iii) $v$ is $M$-universal. 	
  \end{lemma}
 \begin{proof} 
 Since $T$ is $M$-consistent, $T[M\cup \{v\}]$ is acyclic.
 By definition, $v$ can not satisfy more than one property. 
 If $v$ is $M$-universal, then, $v$ is neither $(M,X_i)$-equivalent nor $(M,X_i)$-conflicting for any set $X_i$.
 
 If $v$ is $(M,X_i)$-equivalent to some set $X_i$, then by definition, $v$ is not $M$-universal. In addition, for any set $X_j$, $v$ is not $(M,X_j)$-conflicting as no vertex in $X_i$ is $(M,X_j)$-conflicting.

  Suppose that $v$ is neither $(M,X_i)$-equivalent for any $X_i$ nor $M$-universal. We show that $v$ is $(M,X_i)$-conflicting the first set $X_i$ that contains an out-neighbor $u_i$ of $v$. Suppose that there is an index $j>i$ such that $X_j$ contains an in-neighbor $u_j$ of $v$. Since $j-i\geq 2$, there is an index $i<l<j$ such that $X_l$ lies in the partite set of $T$ different from $X_i\cup X_j$. This gives us a cycle $vu_iu_lu_jv$ where $u_l\in X_l$ contradicting that $T[M\cup \{v\}]$ is acyclic. If every vertex in $X_i$ is an out-neighbor of $v$, then by definition of the canonical sequence, $v$ is $(M,X_{i-1})$-equivalent contradicting the above assumption. Hence, $X_i$ contains an in-neighbor of $v$, thereby proving that $v$ is $(M,X_i)$-conflicting.
 	 \end{proof}
 \begin{definition}[$M$-sequence]
 	Let $T$ be an $M$-consistent bipartite tournament and $(X_1',X_2',\dots)$ be the canonical sequence of $T[M]$. An $M$-sequence $(X_1,Y_1,X_2,Y_2,\dots,X_l,Y_l)$ of $T$ is a sequence of subsets $V(T)$ such that for every index $i$, $X_i$ is the set of all vertices in $V(T)$ that are $(M,X_i')$-equivalent and $Y_i$ is the set of vertices that are $(M,X_i')$-conflicting. In addition, $Y_1$ contains every $M^-$-universal vertex and $Y_l$ contains every $M^+$-universal vertex. For every $i$, the set $X_i\cup Y_i$ is called a \emph{block}, $X_i$ is called the $M$-sub-block and $Y_i$ is called the $\bar{M}$-sub-block.     
 \end{definition}
 \begin{lemma}
 	If $T$ is an $M$-consistent bipartite tournament, then $T$ has a unique $M$-sequence.
 \end{lemma}
 \begin{proof}
 	The existence and the uniqueness of $M$-sequence follows from Lemma \ref{Mconsistent} and the uniqueness of the canonical sequence of $T[M]$. 
 \end{proof}
 As a consequence of Lemma \ref{sequence} and Lemma \ref{Mconsistent}, we get the following lemma.
 \begin{lemma}\label{Msequence}
 	Let $T:=(A,B,E)$ be an $M$-consistent bipartite tournament and $(X_1',X_2',\dots)$ be the canonical sequence of $T[M]$. Let $(X_1,Y_1,X_2,Y_2,\dots)$ be the $M$-sequence of $T$. The following statements hold:
(i) $X_1,Y_1,X_2,Y_2,\dots$ form a partition of $V(T)$, (ii) for each $i$, $X_i'\subseteq X_i$, (iii) for each $i$, $Y_i\cap M=\emptyset$, (iv) for every odd $i$, $X_i\subseteq A, Y_i\subseteq B$ and for every even $i$, $X_i\subseteq B, Y_i\subseteq A$.
 \end{lemma}
\begin{definition}[Refinement]
	A partition $(V_1,V_2,\dots)$ of $U$ is said to be a refinement of another partition $(V_1',V_2',\dots)$ if for every set $V_i$ and $V_j'$, either $V_i\subseteq V_j'$ or $V_i\cap V_j'=\emptyset$.
\end{definition} 
 
 \begin{lemma}\label{refinement}
 	Let $T$ be an acyclic bipartite tournament. Then, for any subset $M\subseteq V(T)$, the canonical sequence of $T$ is a refinement of the $M$-sequence of $T$.
 \end{lemma}
 \begin{proof}
 	Let $(X_1',X_2',\dots)$ be the canonical sequence of $T[M]$ and let $(X_1,Y_1,\dots,X_l,Y_l)$ be the $M$-sequence of $T$. Let $(V_1,V_2\dots)$ be the canonical sequence of $T$. Since each set $V_i$ are twins in $T$, if any vertex in $V_i$ belongs to $X_j$, then every vertex in $V_i$ belongs to $X_j$. If any vertex in $V_i$ is $(M,X_j')$-conflicting, then every vertex in $V_i$ is $(M,X_j')$-conflicting. Hence, $V_i\subseteq Y_j$. If any vertex in $V_i$ is $M$-universal, then every vertex in $V_i$ is $M$-universal. Hence, $V_i\subseteq Y_1$ or $V_i\subseteq Y_l$. The family of sets $V_i$ that contain an $M^-$-universal vertex lie in $Y_1$ and the family of sets $V_i$ that contain an $M^+$-universal vertex lie in $Y_l$.

 \end{proof}
\begin{lemma}\label{adjust}
	Let $T$ and $T\cup \{v\}$ be two $M$-consistent bipartite tournaments and let $(X_1,Y_1,\dots)$ be the $M$-sequence of $T$. Then, there exists an index $i$, such that the $M$-sequence of $T\cup \{v\}$, is either $(X_1,Y_1,\dots,X_i\cup \{v\},\dots)$ or $(X_1,Y_1,\dots,Y_i\cup \{v\},\dots)$.
\end{lemma}
\begin{proof}
	The proof follows from Lemma \ref{Mconsistent}.
\end{proof}
\begin{lemma}
Let $T$ be a bipartite tournament and $H$ be a feedback vertex set of $T$. Let $M\subseteq T-H$ and $P\subseteq H$. Let $(X_1,Y_2,\dots,X_l,Y_l)$ be the $M$-sequence of $T-H$ and $(X_1',Y_1',\dots,X_l',Y_l')$ be the $M$-sequence of $T-P$. Then, for each index $i$, $X_i\subseteq X_i'$ and $Y_i\subseteq Y_i'$. 
\end{lemma}

%%%%%%%%%%%%%%%%%%%%%%%%%%%%%%%%%%%%%%%%%%%%%%%%%%%%%%%%
%%%%%%%%%%%%%%%%%%%%%%%%%%%%%%%%%%%%%%%%%%%%%%%%%%%%%%%%
%%%%%%%%%%%%%%%%%%%%%%%%%%%%%%%%%%%%%%%%%%%%%%%%%%%%%%%%
\section{Constrained Feedback Vertex Set in Bipartite Tournaments}\label{family}

Given a tournament $T$ and an integer $k$, in the first phase of the algorithm for feedback vertex set in tournaments of Kumar and Lokshtanov in \cite{KumarL16}, a family of sub-exponential size of vertex set pairs $(M,P)$ was obtained such that the sought solution $H$ is disjoint from $M$ and contains $P$. The uniqueness of the topological sort of an acyclic tournament implied that any edge going from right to left (referred as \emph{back edge}) \emph{over} an $M$-vertex becomes a conflict edge and must be hit by $H$. The lack of a unique topological sort of an acyclic bipartite tournament breaks down this step as there may be a topological sort of the bipartite tournament such that a back edge is not a conflict edge. We notice that maintaining an addition subset of back edges $F$ that must be hit by $H$ helps in circumventing this issue. With this strategy in mind, we define the \textsc{Constrained Feedback Vertex Set} problem.
\begin{definition}[Constrained Feedback Vertex Set(CFVS)]
Let $T$ be a bipartite tournament with vertex subsets $M,P\subseteq V(T)$, edge set $F\subseteq E(T)$. A feedback vertex set $H$ of $T$ is called $(M,P,F)$-\emph{constrained} if $M\cap H=\emptyset$, $P\subseteq H$ and $H$ is a vertex cover for $F$.
\end{definition}
\smallskip

\noindent
\fbox{\parbox{\textwidth-\fboxsep}{
\textsc{Constrained Feedback Vertex Set} (CFVS)\\
\textbf{Input:} A bipartite tournament, vertex sets $M,P\subseteq V(T)$, edge set $F\subseteq E(T)$ and positive integer $k$.\\
\textbf{Parameter:} $k$\\
\textbf{Task:} determine whether $T$ has an $(M,P,F)$-constrained CFVS $H$ of size at most $k$.
}}

\noindent

\smallskip

In the rest of the paper, we assume that the size of the bipartite tournament is at most $O(k^3)$ as a bi-product of the kernelization algorithm (Lemma \ref{kernel}).
Given a topological sort $\pi$ of an acyclic bipartite tournament $T=(A,B,E)$, we denote $\pi_A$ to be the permutation of $A$ when $\pi$ is restricted to $A$. Similarly, $\pi_B$ denotes the permutation of $B$ when $\pi$ is restricted to $B$. Next, we define a property of a feedback vertex set of a bipartite tournament and while solving for \textsf{BTFVS}, we will look for solutions $H$ that have this property.
\begin{definition}[$M$-homogeneous]
	Let $T$ be a bipartite tournament and $k$ be a positive integer. Let $M\subseteq V(T)$ be a vertex subset such that $T[M]$ is acyclic. A feedback vertex set $H$ of size at most $k$ of $T$ is called $M$-homogeneous if there exists a topological sort $\pi$ of $T-H$ such that every subset of $10\log^3k$ consecutive vertices in $\pi_{A-H}$ or $\pi_{B-H}$ contains a vertex of $M$. 
\end{definition}
The algorithm for CFVS is primarily based on branching and often, given a CFVS instance, a family of CFVS instances with addition properties will be constructed. We abstract it out in the following definition.
\begin{definition}[$\gamma$-reduction]
	A $\gamma$-reduction is an algorithm that given a CFVS instance $(T,M,P,F,k)$ outputs in time $\gamma$ a family $\mathcal{C}:=\{(T,M,P_1,F_1,k),(T,M,P_1,F_1,k),\dots\}$ of size $\gamma$ of CFVS instances such that 
	\begin{description}
		\item[Forward direction] if $(T,M,P,F,k)$ has an $M$-homogeneous $(M,P,F)$-solution, then there exists an instance $(T,M,P_i,F_i,k)\in \mathcal{C}$ that has an $M$-homogeneous $(M,P_i,F_i)$ solution.
		\item[Backward direction] if there exists an instance $(T,M,P_i,F_i,k)\in \mathcal{C}$ that has an $(M,P_i,F_i)$-CFVS solution, then $(T,M,P,F,k)$ has an $(M,P,F)$-solution.	
		\end{description}
\end{definition}
Now, we construct a family of sets $\mathcal{M}$ such that if $(T,k)$ has a solution $H$ of size at most $k$, then there is a set $M\in \mathcal{M}$ such that $H$ is $M$-homogeneous, and hence we can restrict our attention to looking for feedback vertex sets which are $M$-homogeneous for some subset $M$.
\begin{lemma}\label{undeletable1}
There exists an algorithm that given a bipartite tournament $T$ and a positive integer $k$ outputs in time $\gamma$, a family $\mathcal{M}$ of size $\gamma$ of subsets of $V(T)$ for $\gamma=2^{O(\frac{k}{\log k})}$ such that for every feedback vertex set $H$ of size at most $k$ of $T$, there exists $M\in \mathcal{M}$ such that $H$ is $M$-homogeneous.
\end{lemma}
\begin{proof}
	Using $T$ and $k$, we construct $\mathcal{M}$. Let $n=|V(T)|,t=10\log^3k,q=\log^2 k$. As the first step, the algorithm uses Theorem \ref{twise} to construct a family of functions $H_{n,t,q}$ from $[n]$ to $[q]$. Next, the algorithm computes a family $\mathcal{Z}$ of $t$-wise independent subsets of $V(T)$: For each $f\in H_{n,t,q}$, let $Z_f:=\{v_i\in V(T)\mid f(i)=1\}$. Add $Z$ to $\mathcal{Z}$. In the next step, for every subset $Z\in \mathcal{Z}$, compute the family of subsets $\mathcal{M}_Z:=\{M:=Z\setminus \hat{H}\mid \hat{H}\subseteq Z, |\hat{H}|\leq \frac{2k}{\log^2 k}\}$. Finally, output $\mathcal{M}:=\bigcup_{Z\in \mathcal{Z}}\mathcal{M}_{Z}$.

To argue about the correctness of the algorithm, first, we check that the size of $\mathcal{M}$ computed by the above algorithm is consistent with the claim in the lemma. Clearly, $|\mathcal{M}|\leq |H_{n,t,q}|\times |\mathcal{M}_Z|=O(n^t)O((k^3)^{\frac{2k}{\log^2k}})=2^{O(\frac{k}{\log k})}$. We need to show that for every feedback vertex set $H$ of size $k$ and for every topological sort $\pi$ of $T-H$, there exists a function $f\in H_{n,t,q}$ and a set $\hat{H}\subseteq V(T)$ such that $M:=Z\setminus \hat{H}$ satisfies the required properties. Fix a feedback vertex $H$ of size $k$ and a topological sort $\pi$ of $T-H$. First, we prove the following claim:
\begin{claim}
	If we pick $f$ from $H_{n,t,q}$ uniformly at random, then with non-zero probability, the following two events happen:
(i) for every set of $10\log^3 k$ consecutive vertices in $\pi_{A-H}$ or $\pi_{B-H}$, there is a vertex in $Z_f$,
(ii) $|Z_f\cap H|\leq \frac{2k}{\log^2 k}$.
\end{claim}
\begin{proof}
By $t$-wise independence of $H_{n,t,q}$, the probability that no vertex is picked from $t$ consecutive vertices in $\pi_{A-H}$ or $\pi_{B-H}$ is at most $(1-\frac{1}{q})^{t}$. Let $\mathcal{A}_1$ be the event that at least one set of $t$-consecutive vertices either in $\pi_{A-H}$ or in $\pi_{B-H}$ does not contain any vertex from $Z$. Since there at most $n$ sets of $t$-consecutive vertices, by union bound, the probability that event $\mathcal{A}_1$ happens is at most $n\times (1-\frac{1}{q})^{t}\leq Ck^4\times (1-\frac{1}{\log^2 k})^{10\log^3 k}=Ck^4\times \frac{1}{k^{10}}\leq\frac{1}{k^5}$.
Let $\mathcal{A}_2$ be the event that at least $\frac{2k}{\log^2 k}$ vertices of $H$ are in $Z$. The expected number of vertices of $H$ that belong to $Z$ is $k\times \frac{1}{q}=\frac{k}{\log^2 k}$. Therefore, by Markov's inequality, the probability that the event $\mathcal{A}_2$ occurs is at most $\frac{1}{2}$. By union bound the probability that at least one of the events $\mathcal{A}_1$ or $\mathcal{A}_2$ happen is at most $\frac{1}{k^5}+\frac{1}{2}$. Hence, the probability that none of $\mathcal{A}_1$ and $\mathcal{A}_2$ is at least $1-(\frac{1}{k^5}+\frac{1}{2})>0$, thereby implying the claim.
\end{proof}
Hence, the set of functions satisfying the properties in the above claim is non-empty. Let $f$ be such a function. Since, $\mathcal{M}_Z$ is the collection of sets $Z\setminus \hat{H}$ such that $|\hat{H}|\leq \frac{2k}{\log^2 k}$, there exists a choice $\hat{H}$ such that $\hat{H}=Z\cap H$. Hence, $M:=Z\setminus \hat{H}$ satisfies the required properties.

For the runtime of the algorithm, $H_{n,t,q}$ can be constructed in $O(n^t)$ time. For each function $f\in H_{n,t,q}$, the set $Z$ can be obtained in $O(n)$ time. For each $Z$, $\mathcal{M}_Z$ can be obtained in $O(k^{\frac{2k}{\log^2 k}})$ time. Hence, the runtime of the algorithm is $O(n^t)\cdot n\cdot O(k^{\frac{2k}{\log^2 k}})=2^{O(\frac{k}{\log k})}$. 
\end{proof}
\begin{lemma}\label{undeletable}
There exists an algorithm that given a \textsf{BTFVS} instance $(T,k)$ outputs in time $\gamma$, a family $\mathcal{C}:=\{(T,M_1,P_1,\emptyset,k),(T,M_2,P_2,\emptyset,k),\dots\}$ of size $\gamma$ of CFVS instances for $\gamma=2^{O(\frac{k}{\log k})}$ such that
(a) if $(T,k)$ has a solution $H$ of size at most $k$, then $\mathcal{C}$ has a CFVS instance $(T,M,P,\emptyset,k)$ that has an $M$-homogeneous solution of size at most $k$ and (b) if $\mathcal{C}$ has a $(M,P,\emptyset)$-constrained solution, then $(T,k)$ has a feedback vertex set of size at most $k$.
\end{lemma}

\begin{definition}[boundary, vicinity]
	Let $T$ be an acyclic bipartite tournament. Let $M$ be any subset of vertices and $\pi$ be a topological sort of $T$. Let $(X_1,Y_1,\dots)$ be the $M$-sequence of $T$. For any block $X_i\cup Y_i$, the set of vertices in $X_i$ before the first $M$-vertex is called the left boundary of the block and the set of vertices in $X_i$ after the last $M$-vertex is called the right boundary of the block. The vicinity of the block $X_i\cup Y_i$ is the union of the boundaries of $X_i\cup Y_i$, the right boundary of $X_{i-1}\cup Y_{i-1}$, $Y_i$ and the left boundary of $X_{i+1}\cup Y_{i+1}$.  
\end{definition}
\begin{lemma}\label{ratio}
	Let $H$ be an $M$-homogeneous solution for a bipartite tournament $T$. Then, in the $M$-sequence $(X_1,Y_1,X_2,Y_2,\dots)$ of $T-H$, for each $i$, $\frac{|X_i|}{|X_i\cap M|}\leq 20\log^3k$ and $|Y_i|\leq 10\log^3k$. Further, there exists a topological sort of $T-H$ such that the size of each boundary of any block is at most $10\log^3 k$ and the size of the vicinity of any block is at most $30\log^3 k$.
\end{lemma}
\begin{proof}
	The lemma follows immediately after observing that the canonical sequence of $T-H$ is a refinement of $M$-sequence of $T-H$ and any topological sort of $T-H$ preserves the canonical sequence of $T-H$.
\end{proof}
\begin{definition}[Back edge]
	Let $T$ be an $M$-consistent bipartite tournament for some $M\subseteq V(T)$ and $(X_1,Y_1,X_2,Y_2\dots)$ be the $M$-sequence of $T$. An edge $u_iu_j\in E(T)$ is called a \emph{back} edge if  $u_i\in X_i\cup Y_i$, $u_j\in X_j\cup Y_j$ and $i-j\geq 1$. Furthermore, $u_iu_j$ is called \emph{short back} edge if $i-j=1$ and it is called \emph{long back} edge if $i-j\geq 2$. 
\end{definition}
\begin{lemma}\label{longback}
Any feedback vertex set disjoint from $M$ must contain at least one end point of a long back edge.
\end{lemma}

As we know that in the $M$-sequence of $T-H$, there may be back edges. Since $T-H$ is acyclic, these edges do not participate in any cycle. We call them \emph{simple} back edges. But, in the $M$-sequence of $T-P$, we may have back edges that form a cycle with two vertices of $M$ and hence at least one end-point of these edges must belong to $H$. We call them \emph{conflict} back edges. Hence, every back edge that is not a simple back edge is a conflict back edge. By Lemma \ref{longback}, every long back edge is a conflict back edge. The $M$-homogeneity of $H$ and Lemma \ref{ratio} implies the following lemma.
\begin{lemma}\label{simple}
	Let $H$ be an $M$-homogeneous solution for $T$. Then, there exists a permutation of $T-H$ such that the number of simple back edges between any consecutive blocks in the $M$-sequence of $T-H$ is at most $200\log^6 k$.
\end{lemma}
Hence, if in the $M$-sequence of $T-P$, there are more than $200\log^6k$ back edges between any consecutive pair of blocks, then we can branch on the choices of conflict back edges to be hit by $H$. The next definition and lemma captures this intuition.
\begin{definition}[weakly-coupled]
	An instance $(T,M,P,F,k)$ of CFVS is said to be weakly-coupled if in the $M$-sequence 
	 of $T-P$, $F$ is a subset of conflict back edges containing all long back edges such that   
	 the matching in back edges between any pair of consecutive blocks in $T-P-F$ is at most $201\log^{8}k$ . 
\end{definition}
Since we can find a matching in bipartite graphs in polynomial time, it can be checked in polynomial time whether a given CFVS instance $(T,M,P,F,k)$ is weakly-coupled or not.
\begin{lemma}\label{short}
	There exists a $\gamma$-reduction from a CFVS instance $(T,M,P,\emptyset,k)$ to a family $\mathcal{C}_2=\{(T,M,P,F_1,k),(T,M,P,F_2,k)\dots\}$ for $\gamma=2^{O(\frac{k}{\log k})}$ such that every instance in $\mathcal{C}_2$ is weakly-coupled.
\end{lemma}
\begin{definition}[matched]
	An instance $(T,M,P,F,k)$ of CFVS is said to be matched if $F\cap E(T-P)$ forms a matching. 
\end{definition}
Note that it can be checked in polynomial time whether a given CFVS instance $(T,M,P,F,k)$ is matched or not.
\begin{lemma}\label{disjoint}
There exists a $\gamma$-reduction from a weakly-coupled CFVS instance $(T,M,P,F,k)$ to $\mathcal{C}_3:=\{(T,M,P_1,F,k),(T,M,P_2,F,k),\dots\}$ for $\gamma\leq 1.6181^k$ such that  $\mathcal{C}_3$ is weakly-coupled and matched. In addition, for each $|P_i|=s\leq k$, $\mathcal{C}_3$ has at most $1.618^{s}$ CFVS instances.
\end{lemma}
\begin{definition}[LowBlockDegree]
	An instance $(T,M,P,F,k)$ of CFVS is said to be LowBlockDegree if in the $M$-sequence $(X_1,Y_1,X_2,Y_2,\dots)$ of $T-P$, \textsf{long}$(T,M,P)\subseteq F$ and for every set $X_i\cup Y_i$, at most $201\log^{10}k$ edges of $F\setminus E(T-P)$ are incident on $X_i\cup Y_i$.
	\end{definition}
Note that it can be checked in polynomial time whether a given CFVS instance $(T,M,P,F,k)$ is LowBlockDegree or not.
\begin{definition}[$X$-preferred vertex cover]
	Given a bipartite graph $G$ a set of vertices $X\subseteq V(T)$ and a set of edges $Q\subseteq E(G)$ such that $Q$ is a matching in $G$, a minimum vertex cover $C$ of $Q$ is called $X$-vertex cover of $Q$ if for every edge $e\in Q$ such that $e$ has exactly one endpoint in $X$, $C$ contains the endpoint of $e$ in $V(G)\setminus X$.
\end{definition}
Let $T:=(A,B,E)$ be a bipartite tournament and let $X\subseteq A$. Let $\pi:=(v_1,v_2,\dots,v_l)$ be a permutation of $X$. A vertex $v\in B$ is called inconsistent with $\pi$, if there is no index $i$ such that every vertex in $\{v_1,v_2,\dots,v_i\}$ is an in-neighbor of $v$ and every vertex in $\{v_{i+1},v_{i+2,\dots,v_l}\}$ is an out-neighbor of $v$. Given a CFVS instance $(T,M,P,F,k)$, a block in the $M$-sequence of $T-P$ is said to have large conflict edge matching if the block is incident with at least $201\log^{10}k$ edges in $F_1:=F\cap E(T-P)$.
\begin{lemma}\label{disjointmatch}
There exists a $\gamma$-reduction from a weakly-coupled and matched CFVS instance $(T,M,P,F,k)$ to $\mathcal{C}_4:=\{(T,M,P_1,F,k),(T,M,P_2,F,k),\dots\}$ for $\gamma=2^{O(\frac{k}{\log k})}$ such that every instance in $\mathcal{C}_4$	is weakly-coupled, matched and LowBlockDegree.
\end{lemma}
\begin{proof}
Using $M,P,F$, we construct $\mathcal{C}_4$. Start with the $M$-sequence of $T-P$. Let $n=|V(T)|$, $t=\frac{2k}{201\log^{10}k}$, $P':=P$ and $F':=F\cap E(T-P)$. Branch on every family $\mathcal{B}$ of blocks such that $|\mathcal{B}|\leq t$. Branch on every subset $M'$ of size at most $t\cdot 30\log^3 k$. Let $X$ be the union of $M$-sub-blocks and $Y$ be the union of $\bar{M}$-sub-blocks in $\mathcal{B}$. Add every vertex in $Y\setminus M'$ to $P'$. Add every back edge neighbor of $M'$ to $P'$. Branch on every permutation $\pi$ of $M'$. Add every vertex of $X\setminus M'$ not consistent with the permutation $\pi$ to $P'$. Let $E'$ be the set of back edges incident on $X\setminus M'$. Let $G:=(V(T), E')$. Note that $G$ is a bipartite graph. Branch on every minimum vertex cover of $G$ by adding it to $P'$. Add a $\bigcup\mathcal{B}$-preferred cover of conflict edges in $F'$ incident on $(X\cup Y)\setminus P'$ to $P'$. Finally, we add a CFVS instance $(T,M,P',F,k)$ to $\mathcal{C}_4$ if $(T,M,P',F,k)$ is LowBlockDegree.

\emph{Correctness}: First we show that $|\mathcal{C}_4|\leq \gamma$. $|\mathcal{C}_4|$ is bounded by the product of the number of family of blocks $\mathcal{B}$, the number of sets $M'$, the number of permutations of $M'$ and the number of minimum vertex cover of $G$. The number of family of blocks $\mathcal{B}$ is bounded by $n^t$ as the number of blocks can be at most $n$. Similarly, the number of subsets $M'$ is bounded by $n^{t\cdot 30\log^3 k}$. The number of permutations is bounded by $(t\cdot 30\log^3 k)!$. Since, $(T,M,P,F,k)$ is weakly-coupled, the matching on back edges incident on any block is at most $201\log^8 k$. Hence, the size of a maximum matching in $G$ is at most $t\cdot 201\log^8 k$. Hence, the number of minimal vertex cover of $G$ is at most $2^{t\cdot 201\log^8 k}$. Since, $n=|V(T)|=O(k^3)$, after little arithmetic manipulation, we have that $|\mathcal{C}_4|\leq n^t\times n^{t\cdot 30\log^3 k}\times (t\cdot 30\log^3 k)! \times 2^{t\cdot 201\log ^8 k}=2^{O(\frac{k}{\log k})}$.   

 By the definition of the family $\mathcal{C}_4$ and of $\gamma$-reduction, the backward direction is immediate. For the forward direction, let $(T,M,P,F,k)$ be a weakly-coupled and matched CFVS instance and let $H$ be an $M$-homogeneous solution of $(T,M,P,F,k)$. It is sufficient to show that $\mathcal{C}_4$ has an instance $(T,M,P',F,k)$ such that $P'\subseteq H$.
 
 Consider the $M$-sequence of $T-P$. Fix a permutation $\sigma$ of $T-H$. Consider a permutation $\sigma'$ of vertices in $T-P$ whose restriction to $T-H$ is $\sigma$. Let $\mathcal{B}$ be the family of blocks with \emph{very large} matching in the set of conflict edges $F'$. Since $|H|\leq k$, the size of $\mathcal{B}$ is less than $t=\frac{2k}{201\log^{10}k}$. Since the size of vicinity of any block is at most $30\log^3 k$, at most $t\cdot 30\log^3 k$ vertices form the vicinity $M'$ of blocks in $\mathcal{B}$. Let $X$ be the union of $M$-sub-blocks and $Y$ be the union of $\bar{M}$-sub-blocks in $\mathcal{B}$. Then, vertices in $Y\setminus M'$ belong to $H$. Since, $M'$ is the vicinity of the blocks, every back edge incident on $M'$ is a conflict edge. Hence, the back edge neighbor of $M'$ belongs to $H$. For the same reason, the set of back edges $E'$ incident on $X\setminus M'$ are conflict edges and $H$ contains a minimum cover of $E'$. 
 As $M'\cap H=\emptyset$, any vertex inconsistent with $\sigma_{M'}$ also belongs to $H$. Now, every block in $\mathcal{B}$ is incident with conflict edges belong to $F$ only which are disjoint, we can greedily include a vertex cover of these edges by preferring to pick the conflict edge neighbor of $\bigcup \mathcal{B}$ into $H$. This implies that every block in $\mathcal{B}$ after removing $P'$ is not incident with any conflict edge and hence $(T,M,P',F,k)$ is LowBlockDegree. Since $P'$ includes all possibilities of the above choices, there is an instance $(T,M,P',F,k)$ in $\mathcal{C}_4$ that satisfies the required properties.
\end{proof}
\begin{definition}[Regular]
	An instance $(T,M,P,F,k)$ of CFVS is said to be regular if in the $M$-sequence $(X_1,Y_1,X_2,Y_2,\dots)$ of $T-P$, 
 for every set $X_i$ of size at least $10\log^5 k$, there are at least $\frac{|X_i|}{10\log^5 k}$ vertices in $M$ and $|Y_i|\leq 10\log^5 k$. 
\end{definition}
Note that it can be checked in polynomial time whether a given CFVS instance $(T,M,P,F,k)$ is regular or not. 
Let $\mathcal{L}$ be a function such that given a CFVS instance $(T,M,P,F,k)$ outputs the family of sets of vertices which is the union of all sets $X_i$ and $Y_j$ in the $M$-sequence of $T-P$ such that $\frac{|X_i|}{m_i}\geq 10\log^5 k$ where $m_i=|X_i\cap M|$ and $|Y_j|\geq 10\log^5 k$.
\begin{lemma}\label{blocks}
	There exists a $\gamma$-reduction from a CFVS instance $(T,M,P,F,k)$ to a family $\mathcal{C}_1:=\{(T,M,P_1,F,k),(T,M,P_2,F,k),\dots\}$ of CFVS instances for $\gamma=2^{O(\frac{k}{\log k})}$ such that every instance in $\mathcal{C}_1$ is regular.
\end{lemma}
As noted before BTFVS instance $(T,k)$ is equivalent to CFVS instance $(T,\emptyset,\emptyset,\emptyset,k)$, we combine the results in the above Lemmas (abusing the notation slightly).
\begin{lemma}\label{reduction}
	There is a $\gamma$-reduction from a BTFVS instance $(T,k)$  to a CFVS family $\mathcal{C}'$ for $\gamma\leq 1.6181^k$ such that every instance in $\mathcal{C}'$ is regular, weakly-coupled, matched and LowBlockDegree. In addition, for each $|P_2|=s\leq k$, $\mathcal{C}'$ has at most $1.618^{s}$ CFVS instances.
\end{lemma}
We redefine the $d$-\textsc{Feedback Vertex Cover} defined in \cite{KumarL16} with a slight modification. Let $d,f$ and $t$ be positive integers. Consider a class of mixed graphs $\mathcal{G}(d,f,t)$ in which each member is a mixed multigraph $\mathcal{T}$ with
the vertex set $V(\mathcal{T})$ partitioned into vertex sets $V_1,V_2,\dots,V_t$ and an undirected edge set $\mathcal{E}(\mathcal{T})\subseteq \bigcup_{i<j}V_i\times V_j$ such that for each $i\in [t]$,
(a) $\mathcal{T}[V_i]$ is a bipartite tournament, (b) the size of the feedback vertex set $H_i$ for $\mathcal{T}[V_i]$ is at least $f$ and at most $4f$, and (c) \textsf{deg}$_{\mathcal{E}}(V_i)\leq d$.

Given a mixed multigraph $\mathcal{T}\in \mathcal{G}(d,f,t)$, a positive integer $k$, determine whether there exists a set $H\subseteq V(\mathcal{T})$ such that $|H|\leq k$ and $\mathcal{T}-H$ contains no undirected edges and is acyclic. If $\mathcal{E}(\mathcal{T})$ is disjoint, we call the problem as \textsc{Disjoint Feedback Vertex Cover}.
\begin{lemma}\label{part}
	There exists a polynomial time algorithm that given a CFVS instance $(T,M,P_2,F,k)$ that is regular, weakly-coupled, matched and LowBlockDegree outputs a partition $(V_1,V_2,\dots,V_t)$ of $V(T)\setminus P$ such that $t\leq \frac{k}{201\log^{12}k}$ and for each $i\in [t]$ $V_i$ is a union of consecutive blocks in the $M$-sequence of $T-P_2$ and at least one of these hold 
	\begin{itemize} 
		\item the size of feedback vertex set of $T[V_i]$ is at least $f=201\log^{12}k$ and at most $804\log^{12}k$, 
		\item at least $200\log^{12}k$ and at most $201\log^{12}k$ edges in $F\cap E(T-P_2)$ are incident on $V_i$. 
	\end{itemize}
\end{lemma}
\begin{proof}
	Let $(X_1,Y_1\dots)$ be the $M$-sequence of $T-P_2$. Consider the sequence of blocks $(Z_1,Z_2\dots)$ such that for each $i$, $Z_i:=X_i\cup Y_i$. Obtain the sequence $i_1=1<i_2<\dots$ of indices such that $V_j:=\bigcup_{i=i_j}^{i_{j+1}-1}Z_i$ as follows: for each $j$, keep including $Z_i$ for $i\geq i_j$ into $V_j$ and stop the moment at least one of the above conditions hold. To check the size of feedback vertex set in $T[V_j]$ use the approximation algorithm in Lemma \ref{approximation} i.e. check if Lemma \ref{approximation} outputs a feedback vertex set for $T[V_j]$ of size less than $4f$.
	
	Since by regularity, the feedback vertex set of any block is at most $10\log^5k$ and since the CFVS instance is weakly-coupled, the size of a maximum matching on back edges between any consecutive blocks is at most $201\log^{8}k$. Since the CFVS instance is matched and LowBlockDegree, the size of maximum matching in conflict edges is at most $201\log^{10}k$. Hence, including any block into a set $V_i$ increases the size of the feedback vertex set of $T[V_i]$ by at most $10\log^5k+201\log^{10}k$. At the same time, the degree of $V_i$ can increase by at most $201\log^{10}k$. Hence, the above algorithm outputs the required partition. Note that edges in $F\cap E(T-P_2)$ form a matching. Hence, $t\leq \frac{k}{201\log^{12}k}$.
\end{proof}
\begin{definition}[decoupled]
An instance $(T,M,P,F,k)$ of CFVS is said to be decoupled if there is a partition $(V_1,V_2,\dots,V_t)$ of $V(T)\setminus P$ such that $t\leq \frac{k}{201\log^{12}k}$ and for each $i\in [t]$
(a) $V_i$ is a union of consecutive blocks in the $M$-sequence of $T-P$, (b) the size of feedback vertex set of $T[V_i]$ is at least $f=201\log^{12}k$ and at most $804\log^{12}k$, or at least $200\log^{12}k$ and at most $d=201\log^{12}k$ edges in $F\cap E(T-P)$ are incident on $V_i$. (c) $F$ contains short conflict edges between any pair of sets $V_i$ and $V_j$.
\end{definition}
Note that it can be checked in polynomial time whether a given CFVS instance $(T,M,P,F,k)$ is decoupled or not.
\begin{lemma}\label{super}
	There exists a $\gamma$-reduction from a regular, weakly-coupled, and matched CFVS instance $(T,M,P_2,F,k)$ to a family $\mathcal{C}_6$ for $\gamma=2^{O(\frac{k}{\log k})}$ such that every instance in $\mathcal{C}_6$ is regular, weakly-coupled, matched, LowBlockDegree and decoupled. 

\end{lemma}
\begin{proof}
	 Given $(T,M,P_2,F,k)$, we construct the family $\mathcal{C}_6$. Using the algorithm of Lemma \ref{part}, we construct the partition $(V_1,V_2,\dots,V_t)$ of $V(T)\setminus P_2$. For each $V_i$, let $E_i$ be the set of back edges incident on $V_i$ from $V(T)\setminus (P_2\cup V_{i})$. Let $J:=\bigcup_{V_i}E_i$ be the union of such back edges. Now, we guess the subset $B$ of back edges that are not hit by the required feedback vertex set. For every subset $B\subseteq J$ of size at most $2\cdot t\cdot 200\log^6 k$, let $J_B=J\setminus B$. We require that the feedback vertex set hits at least one end point of every edge in $J_B$. Let $D$ be the vertex cover of $J_B$. For every subset $C\subseteq D$, define $P_C:=C\cup N_{J_B}(D\setminus C)$. For each $P_C$, we add the CFVC instance $(T,M,P_3,F,k)$ where $P_3:=P_2\cup P_C$ into $\mathcal{C}_6$ if $(T,M,P_3,F,k)$ is regular, weakly-coupled, matched, LowBlockDegree and decoupled. 

The backward direction is trivial. For the forward direction, let $H$ be an $M$-homogeneous $(M,P,F)$-CFVS solution.
Observe that all the above algorithm does is consider all possibilities via which $H$ may hit the back edges between $T[V_i\setminus P_2]$ and $T[V_i\setminus P_2]$ for any $i,j$. The number of choices of sets $B$ is at most $(k^6)^{2\cdot t\cdot 200\log^6 k}=2^{O(\frac{k}{\log k})}$. Note that in the $M$-sequence of $T-P_2$, the matching on short back edges between any pair of consecutive blocks is at most $201\log^{10}k$. Hence, the vertex cover of these back edges is at most $201\log^{10}k$. Since the number of sets in the partition $(V_1,V_2,\dots)$ is at most $\frac{k}{f}$, the size of the total matching on short back edges $J$ is at most $g=201\log^{10}k\times \frac{k}{f}$. Hence, the number of choices for $C$ is at most $2^g=2^{O(\frac{k}{\log k})}$. Hence, $\gamma=2^{O(\frac{k}{\log k})}\times 2^{O(\frac{k}{\log k})}=2^{O(\frac{k}{\log k})}$.
\end{proof}
\begin{lemma}\label{dfvc}
	There is a polynomial time reduction from a CFVS instance $(T,M,P_2,F,k)$ that is regular, weakly-coupled, matched, LowBlockDegree and decoupled to an instance of \textsc{Disjoint Feedback Vertex Cover} $(\mathcal{T},k')$ for $k'=k-|P_2|$. 	
\end{lemma}
\begin{proof}
 Given $(T,M,P_2,F,k)$, construct the DFVS instance with vertex set $V(T)\setminus (M\cup P_2)$ and make the edges in $F\setminus E(T-P_2)$ between any two sets $V_i$ and $V_j$ undirected. For any solution $H$ for $(T,M,P_2,F,k)$, $H\setminus P_2$ is a feedback vertex set of $T-P_2$ that hits $F\setminus E(T-P_2)$. Hence, $H\setminus P_2$ is a feedback vertex cover for $(\mathcal{T},k')$ for $k'=k-|P_2|$. In the backward direction, a solution $S$ for $(\mathcal{T},k')$ hits $F\setminus E(T-P_2)$ and is disjoint from $M$. Hence, $S\cup P_2$ is a solution for $(T,M,P_2,F,k)$. 
\end{proof}
At this point, we can use the following lemma from \cite{KumarL16} with the only difference being in the base case as we have a bipartite tournament instead of a \emph{supertournament}. We replace the naive $3^k$ algorithm by $4^k$ algorithm to find a feedback vertex for each of $T[V_i]$. Note that bounding the size of feedback vertex set in each of $T[V_i]$ to $O(\log^{12}k)$ and the number of $V_i$'s to at most $O(\frac{k}{\log^{12}k})$ implies that the maximum time spent in solving the base cases is at most $O(\frac{k}{\log^{12}k})\cdot 2^{ O(\log^{12}k)}$. 
\begin{lemma}\cite{KumarL16}\label{dfvcalgo}
There exists an algorithm running in $1.5874^s\cdot 2^{O(df\log k+d\log s)}\cdot n^{O(1)}$ time which finds an optimal feedback vertex cover in a mixed multigraph $\mathcal{T}\in \mathcal{G}(d,f,t)$ in which the undirected edge set $\mathcal{E}(\mathcal{T})$ is disjoint and $|\mathcal{E}(\mathcal{T})|=s$. 
\end{lemma}
\begin{theorem}
There exists an algorithm for \textsf{BTFVS} running in $1.6181^k + n^{O(1)}$ time.
\end{theorem}
\begin{proof} Using the algorithm of Lemma \ref{reduction}, construct the family $\mathcal{C}_5$ of CFVS instances. For each instance $(T,M,P_2,F,k)\in \mathcal{C}_5$, using the algorithm of Lemma \ref{super} construct the family $\mathcal{C}_6$ of CFVS instances. Then for each CFVS instance $(T,M,P,F,k)$ using Lemma \ref{dfvc}, construct the DFVC instance $(\mathcal{T},k-|P|)$ which is solved using the algorithm of Lemma \ref{dfvcalgo}. If for any instance, the algorithm of Lemma \ref{dfvcalgo} outputs a solution set $S$ of size at most $k-|P|$, then we output \textsc{yes}, otherwise output \textsc{no}.
	
	The correctness of the algorithm follows from the correctness of the algorithms in the Lemma \ref{reduction}, \ref{super}, \ref{dfvc} and \ref{dfvcalgo}.
	The runtime of the algorithm is upper bounded by $\sum\limits_{s=1}^{k}1.618^{k-s}\times 1.5874^{s}\cdot 2^{O(df\log k+d\log s)}\cdot n^{O(1)}\leq 1.6181^k\cdot 2^{O(d^2+d\log k)}\cdot n^{O(1)}$.
\end{proof}
\begin{proposition}\cite{FominGLS16}
If there exists a parameterized algorithm for any vertex deletion problem into a hereditary graph class with running time $c^kn^{O(1)}$, then there exists an exact-exponential-time algorithm for the problem with running time $(2-\frac{1}{c})^{n+o(n)}n^{O(1)}$.
\end{proposition}
The above proposition immediately implies the following theorem.
\begin{theorem}
	There exists an algorithm for \textsf{BTFVS} running in $1.3820^n$ time.
\end{theorem}
%------------------------------------------------------------------

%%%%%%%%%%%%%%%%%%%%%%%%%%%%%%%%%%%%%%%%%%%%%%%%%%%%%%%%
%%%%%%%%%%%%%%%%%%%%%%%%%%%%%%%%%%%%%%%%%%%%%%%%%%%%%%%%
%%%%%%%%%%%%%%%%%%%%%%%%%%%%%%%%%%%%%%%%%%%%%%%%%%%%%%%%

\subparagraph*{Acknowledgements} 
The research leading to these results has received funding from the European Research Council under the European Union's Seventh Framework Programme (FP7/2007-2013) / ERC grant agreement no. 306992 and the Beating Hardness by Pre-processing grant funded by the Bergen Research Foundation.
%-------------------------------------------------
\appendix
\nocite{DBLP:journals/ipl/GutinJ14}
\bibliography{110-Kumar}{}
%!tex root=btfvs.tex
\section*{Appendix}%\todo{permute lemma order}
In this paper, we work with graphs that do not contain any self loops. 
A {\em multigraph} is a graph that may contain more than one edge between the same pair of vertices. A graph is {\em mixed} if it can contain both directed and undirected edges. We will be working with mixed multigraphs; graphs that contain both directed and undirected edges, and where two vertices may have several edges between them. 

When working with a mixed multigraph $G$ we use $V(G)$ to denote the vertex set, $E(G)$ to denote the set of directed edges, and $\mathcal{E}(G)$ to denote the set of undirected edges of $G$. A directed edge from $u$ to $v$ is denoted by $uv$.  

\smallskip
\noindent
{\bf Graph Notation.} In a directed graph $D$, the set of \emph{out-neighbors} of a vertex $v$ is defined as $N^+(v):=\{u|vu\in E(D)\}$. Similarly, the set of \emph{in-neighbors} of a vertex $v$ is defined as $N^-(v):=\{u|uv\in E(D)\}$. A \emph{square} in a directed graph is a directed cycle of length $4$. Note that in this paper, whenever the term square is used it refers to a directed square. A pair of vertices $u,v$ are called \emph{false twins} if $uv\notin E(D),vu\notin E(D)$ and $N^+(u)=N^+(v), N^-(u)=N^-(v)$. A \emph{topological sort} of a directed graph $D$ is a permutation $\pi:V(D)\to [n]$ of the vertices of the graph such that for all edges $uv\in E(D)$, $\pi(u)<\pi(v)$. Such a permutation exists for a directed graph if and only if the directed graph is acyclic. For an acyclic tournament, the topological sort is unique. For an acyclic bipartite tournament, the topological sort is unique up to permutation of false twins. 

For a graph or multigraph $G$ and vertex $v$, $G - v$ denotes the graph obtained from $G$ by deleting $v$ and all edges incident to $v$. For a vertex set $S$, $G - S$ denotes the graph obtained from $G$ by deleting all vertices in $S$ and all edges incident to them.

For any set of edges $C$ (directed or undirected) and set of vertices $X$, the set $V_C(X)$ represents the subset of vertices of $X$ which are incident on an edge in $C$. For a vertex $v\in V(G)$, the set $N_C(v)$ represents the set of vertices $w\in V(G)$ such that there is an undirected edge $wv\in C$. We define a vertex cover $S$ for a set of edges $F$ to be a set of endpoints of $F$ such that every edge has at least one endpoint in $S$. For a bipartite graph $T:=(A,B,E)$, $A$ and $B$ are called the partite sets of $T$.

\smallskip
\noindent
{\bf Fixed Parameter Tractability.} A {\em parameterized problem} $\Pi$ is a subset of $\Sigma^* \times \mathbb{N}$. A parameterized problem $\Pi$ is said to be \emph{fixed parameter tractable}(\textsc{FPT}) if there exists an algorithm that takes as input an instance $(I, k)$ and decides whether $(I, k) \in \Pi$ in time $f(k)\cdot n^c$, where $n$ is the length of the string $I$, $f(k)$ is a computable function depending only on $k$ and $c$ is a constant independent of $n$ and $k$. 

A \emph{kernel} for a parameterized problem $\Pi$ is an algorithm that given an instance $(T,k)$ runs in time polynomial in $|T|$, and outputs an instance $(T',k')$ such that $|T'|,k' \leq g(k)$ for a computable function $g$ and $(T,k) \in \Pi$ if and only if $(T',k') \in \Pi$. For a comprehensive introduction to \textsc{FPT} algorithms and kernels, we refer to the book by Cygan et al.~\cite{pc_book}.

\smallskip
\noindent
\textbf{Lemma~\ref{undeletable} (Restated)}
{\em There exists an algorithm that given a \textsf{BTFVS} instance $(T,k)$ outputs in time $\gamma$, a family $\mathcal{C}:=\{(T,M_1,P_1,\emptyset,k),(T,M_2,P_2,\emptyset,k),\dots\}$ of size $\gamma$ of CFVS instances for $\gamma=2^{O(\frac{k}{\log k})}$ such that
\begin{itemize}
	\item if $(T,k)$ has a feedback vertex set $H$ of size at most $k$, then $\mathcal{C}$ has a CFVS instance $(T,M,P,\emptyset,k)$ that has an $M$-homogeneous solution of size at most $k$ and
	\item if $\mathcal{C}$ has a $(M,P,\emptyset)$-constrained solution, then $(T,k)$ has a feedback vertex set of size at most $k$.
\end{itemize}} 
\begin{proof}
Given $(T,k)$, we use the algorithm of Lemma \ref{undeletable1} with $T,k$ as input and obtain the family of sets $\mathcal{M}$. For each set $M\in \mathcal{M}$, we add a CFVS instance $(T,M,P,\emptyset,k)$ in $\mathcal{C}$ where $P$ is the set of vertices in $V(T)\setminus M$ that form a cycle of length $4$ with $M$.
 For the forward direction, if $(T,k)$ has a solution $H$ of size at most $k$, then by Lemma \ref{undeletable1}, there exists a set $M\in \mathcal{M}$ such that $H$ is $M$-homogeneous. Since, $P$ is the set of vertices that form a cycle with $M$ and $M\cap H=\emptyset$, $P\subseteq H$. Hence, $(T,M,P,\emptyset,k)\in \mathcal{C}$ has an $M$-homogeneous solution. The backward direction immediately follows from the construction of $\mathcal{C}$.
\end{proof}
\textbf{Lemma~\ref{longback} (Restated)}
{ \em Any feedback vertex set disjoint from $M$ must contain at least one end point of a long back edge.}
%\end{lemma}
\begin{proof}
	Let $u_ju_i$ be a long back edge and $i<j$ and $u_i\in X_i,u_j\in X_j$. Then, there are two vertices $u_l\in X_l$ and $u_{l+1}\in X_{l+1}$ in $M$ such that $i<l<l+1<j$. This creates the cycle $u_iu_lu_{l+1}u_ju_i$. Since $u_l$ and $u_{l+1}$ are undeletable, the feedback vertex set must contain at least one of $u_i$ and $u_j$.
	
	Now, consider the case when $u_i\in Y_i$ and $u_j\in X_j$. Since $u_i$ is $(M,X_i)$-conflicting, there is an out neighbor $u\in X_i$ of $u_i$. Again, since $j-i\geq 2$, there is a set $X_l$ such that $i<l<j$ and we get a cycle $u_iuu_lu_ju$ where $u_l\in X_l$. The case when $u_i\in Y_i$ and $u_j\in Y_j$ is similar. 
\end{proof}

Let $T$ be a bipartite tournament such that $T-P$ is $M$-consistent for some sets $M,P\subseteq V(T)$. Let $\mathcal{L}'$ be a function such that given an $M$-consistent bipartite tournament $T$ for some set $M\subseteq V(T)$ and an integer $k$, outputs the set of short back edges in the $M$-sequence $(X_1,Y_1,\dots)$ of $T$ which is the union of all sets of back edges $E_{i,i+1}$ between $X_i\cup Y_i$ and $X_{i+1}\cup Y_{i+1}$ such that the size of matching in the bipartite graph $(X_i\cup Y_{i+1},X_{i+1}\cup Y_i,E_{i,i+1})$ is at least $201\log^{8} k$. Let \textsf{long}$(T,M,P)$ denote the set of long back edges in $T-P$.

\medskip
\noindent
\textbf{Lemma~\ref{short} (Restated)}
{\em	There exists a $\gamma$-reduction from a CFVS instance $(T,M,P,\emptyset,k)$ to a family $\mathcal{C}_2=\{(T,M,P,F_1,k),(T,M,P,F_2,k)\dots\}$ for $\gamma=2^{O(\frac{k}{\log k})}$ such that every instance in $\mathcal{C}_2$ is weakly-coupled.}
%\end{lemma}
\begin{proof}
We construct $\mathcal{C}_2$ as follows. For each $B\subseteq \mathcal{L}'(T-P,M,k)$ such that $|B|\leq \frac{2k}{\log^2 k}$ output a set $F_B:=\mathcal{L}'(T,M,k)\setminus B\cup$\textsf{long}$(T,M,P)$. For each set $F_B$, add the instance $(T,M,P,F_B,k)$ in $\mathcal{C}_2$ if $(T,M,P,F_B,k)$ weakly-coupled. 

By the definition of $\gamma$-reduction and the construction of $\mathcal{C}_1$, the backward direction is trivial. Now we consider the forward direction. Let $H$ be an $M$-homogeneous $(M,P,\emptyset,k)$-constrained solution for $(T,M,P,\emptyset,k)$ and $(X_1,Y_1,\dots)$ be the $M$-sequence of $T-P$. It is sufficient to show that there is a CFVS instance $(T,M,P,F,k)$ such that $H$ is a vertex cover of $F$. A pair of consecutive blocks is said to have \emph{large back} edge matching if the size of a matching in the set of back edges between them is at least $201\log^{8} k$. Fix a permutation $\sigma$ of $T-H$ and choose any permutation $\sigma'$ of $T-P$ such that $\sigma'_{T-H}$ is $\sigma$. By Lemma \ref{simple}, we have that at most $200\log^6k$ edges are simple back edges between any pair of consecutive blocks in the $M$-sequence of $T-H$. Rest of the back edges are conflict back edges and must be hit by $H$. If the short back edge matching is large between a pair of blocks, then at least $201\log^{8} k-200\log^{6}k\geq 200\log^{8}k$ of them are conflict back edges. Hence, the number of set pairs with large short edge matching can be at most $\frac{k}{200\log^{8}k}$. This implies that at most $2\times\frac{k}{200\log^{8}k}\times 200\log^{6}k=\frac{2k}{\log^2 k}$ edges are simple back edges. Since the algorithm loops over all choices of subsets $B\subseteq \mathcal{L}'(T,M,k)$, $|B|\leq \frac{2k}{\log^2 k}$, $\mathcal{C}_2$ contains an instance with the required properties. 

Moreover, $|\mathcal{C}_2|$ is bounded by the number of subsets $B$. Now $|\mathcal{L}'(T,M,k)|\leq |V(T)|^2$ which implies the number of subsets $B$ is at most $(k^6)^{\frac{2k}{\log^2k}}=2^{6\log k\times \frac{2k}{\log^2k}}=2^{O(\frac{k}{\log k})}$. Hence, $|\mathcal{C}_2|=2^{O(\frac{k}{\log k})}$
\end{proof}
\textbf{Lemma~\ref{disjoint} (Restated)}
{\em There exists a $\gamma$-reduction from a weakly-coupled CFVS instance $(T,M,P,F,k)$ to $\mathcal{C}_3:=\{(T,M,P_1,F,k),(T,M,P_2,F,k),\dots\}$ for $\gamma\leq 1.6181^k$ such that  $\mathcal{C}_3$ is weakly-coupled and matched. In addition, for each $|P_i|=s\leq k$, $\mathcal{C}_3$ has at most $1.618^{s}$ CFVS instances.}
%\end{lemma}
\begin{proof}
We construct the family $\mathcal{C}_3$ using a branching algorithm. Consider the graph $G:=(V(T)\setminus P,F\cap E(T-P))$. Start with $k'=k$, $P':=P$ and $F':=E(G)$. In each branch node, the sets $P',F'$ are updated and finally for each leaf node in the branch tree, the corresponding instance $(T,M,P',F,k)$ is returned. As long as there is a vertex $v\in V(G)$ of degree at least $2$ and $k'>0$, branch by considering both the possibilities: $v\in H$ or $v\notin H$. In the branch in which $v$ is picked, decrease $k'$ by 1 and update $P'=P'\cup \{v\}$ and $F'=F'\setminus E(v)$. In the other branch, $N(v)$ is added to $P'$, $E(N(v))$ is removed from $F'$ and $k'$ is decreased by $|E(N(v))|$. The algorithm stops branching further in a branch in which either $k'<0$ or $k'>0$ and for every vertex $v$, degree of $v$ is at most 1. In the case that $k'<0$ or $|F'|> k$, the algorithm terminates the branch without returning any instance and moves on to other branches. Any returned instance $(T,M,P',F,k)$ is added to $\mathcal{C}_3$ if the instance is regular, weakly-coupled and matched. 
	
	Again, the definition of the $\gamma$-reduction and above construction of $\mathcal{C}_3$, ensures the backward direction. Now we consider the forward direction. Let $(T,M,P,F,k)$ be weakly-coupled and $H$ be an $M$-homogeneous $(M,P,F)$-constrained solution of $T$. Since, the above branching algorithm adds an instance into $\mathcal{C}_3$ with $P'$ containing $P$ such that $F\cup E(T-P')$ forms a matching, if $\mathcal{C}_3$ is non-empty, all instances in it are weakly-coupled and matched. Since $H$ hits $F$, there is a subset $\mathcal{P''}$ of $H$, such that $F$ forms a matching in $T-(P\cup P'')$. Since the above algorithm via branching considers all possible subsets $P'$ containing $P$ that \emph{make} $F$ disjoint in some branch $P'\subseteq P$ implying that $\mathcal{C}_3$ contains an $M$-homogeneous $(M,P',F)$-constrained solution.
	
	Now we argue about $\gamma$ and the number of instances. Let $s$ denote the size of $P'$ in any instance $(T,M,P',F,k)$. Since, $H$ must hit $F'$ and $F'$ are disjoint $s\leq k$. As $P'\subseteq H$, $|F'|\leq k-s$. The recurrence relation for bounding the number of leaves with $|F'|=k-s$ in the branch tree of the above algorithm is given by:
	$$g_s(k)\leq g_s(k-1)+g_s(k-2)$$
	which solves to $g_s(k)\leq 1.618^{s}$ as $g_s(k)\leq 1$ for $k=s$.
\end{proof}
\textbf{Lemma~\ref{blocks} (Restated)}
{\em	There exists a $\gamma$-reduction from a CFVS instance $(T,M,P,F,k)$ to a family $\mathcal{C}_1:=\{(T,M,P_1,F,k),(T,M,P_2,F,k),\dots\}$ of CFVS instances for $\gamma=2^{O(\frac{k}{\log k})}$ such that every instance in $\mathcal{C}_1$ is regular.}
%\end{lemma}
\begin{proof} 
We construct $\mathcal{C}_1$ as follows. Compute the sets $\mathcal{L}(T,M,k)$. For each $B\subseteq \mathcal{L}(T,M,k)$ such that $|B|\leq \frac{2k}{\log^2 k}$, output a pair of sets $(M, P')=(M, P\cup \mathcal{L}(T,M,k)\setminus B)$. For each pair $(M,P')$, add a CFVS instance $(T,M,P',\emptyset,k)$ in $\mathcal{C}_1$ if $(T,M,P',\emptyset,k)$ is regular. 

By the definition of $\gamma$-reduction and the construction of $\mathcal{C}_1$, the backward direction is trivial. For the forward direction, let $H$ be an $M$-homogeneous $(M,P,F)$-CFVS solution of $(T,M,P,F,k)$. 

A set $X_i$ in the $M$-sequence of $T-P$ is called \emph{large} if the ratio $\frac{|X_i|}{m_i}$ is at least $10\log^5k$. Similarly, $Y_i$ is large if $|Y_i|\geq 10\log^5 k$. From each large set $X_i$ at least $10m_i\log^5k-10m_i\log^3k$ vertices belong to $H$. Similarly, from each large $Y_i$, at least $10\log^5 k-10\log^3 k$ belongs to $H$.  Hence, if $t$ is the total number of $M$-vertices in the union of large sets, then in total at most $\frac{k}{10t\log^5k-10t\log^3k}\times 10t\log^3k\leq \frac{2k}{\log^2k}$ vertices from the union of \emph{large} sets in the $M$-sequence of $T-P$ do not belong to $H$.  Since the algorithm loops over all choices of subsets $B\subseteq \mathcal{L}(T,M,k)$, $|B|\leq \frac{2k}{\log^2 k}$, $\mathcal{C}_1$ contains an instance $(T,M,P',F,k)$ satisfying the required properties.

Moreover, $|\mathcal{C}_1|$ is bounded by the number of subsets $B$. Now $|\mathcal{L}(T,M,k)|\leq |V(T)|$ which implies the number of subsets $B$ is at most $(k^3)^{\frac{2k}{\log^2k}}=2^{3\log k\times \frac{k}{\log^2k}}=2^{O(\frac{k}{\log k})}$. 
\end{proof}
\textbf{Lemma~\ref{reduction} (Restated)}
{\em	There is a $\gamma$-reduction from a BTFVS instance $(T,k)$  to a CFVS family $\mathcal{C}'$ for $\gamma\leq 1.6181^k$ such that every instance in $\mathcal{C}'$ is regular, weakly-coupled, matched and LowBlockDegree. In addition, for each $|P_2|=s\leq k$, $\mathcal{C}'$ has at most $1.618^{s}$ CFVS instances.}
%\end{lemma}
\begin{proof}
  Given the \textsf{BTFVS} instance $(T,k)$, run the algorithm of Lemma \ref{undeletable} to obtain the CFVS family $\mathcal{C}$. For each instance $(T,M,P,\emptyset,k)\in \mathcal{C}$, run the algorithm of Lemma \ref{blocks} to obtain the CFVS family $\mathcal{C}_1$. For each instance $(T,M,P,\emptyset,k)\in \mathcal{C}_1$, run the algorithm of Lemma \ref{short} to obtain the CFVS family $\mathcal{C}_2$. For each instance $(T,M,P,F,k)\in \mathcal{C}_2$, run the algorithm of Lemma \ref{disjoint} to obtain the CFVS family $\mathcal{C}_3$. For each instance $(T,M,P,F,k)\in \mathcal{C}_3$, run the algorithm of Lemma \ref{disjointmatch} to obtain the CFVS family $\mathcal{C}_5(T,M,P_2,F,k)$. $\mathcal{C}'$ is the union of these families.
	
	The correctness and the runtime follow from Lemmas \ref{undeletable}, \ref{blocks}, \ref{short}, \ref{disjoint} and \ref{disjointmatch}.
\end{proof}

\end{document}